\documentclass[submission,copyright,creativecommons]{eptcs}
\providecommand{\event}{ALP4IoT 2017} 
\usepackage{breakurl}             
\usepackage[english]{babel}
\usepackage{underscore}           

\usepackage{amsmath}
\usepackage{amsfonts}
\usepackage{amssymb}
\usepackage{amsthm}
\usepackage{calc}
\usepackage{graphicx}
\graphicspath{ {./figs/} }

\usepackage{array}
\usepackage{booktabs}
\usepackage{listings}
\usepackage{tikz}
\usetikzlibrary{arrows,backgrounds,patterns,matrix,shapes,fit,calc,shadows,plotmarks,decorations.pathmorphing,decorations.pathreplacing,automata,positioning}
\usepackage{todonotes}
\usepackage{subcaption} 
\usepackage{algorithm}
\usepackage{algpseudocode}
\usepackage{booktabs}
\usepackage{logic}
\usepackage{topologies}

\newtheorem{lemma}{Lemma}

\newtheorem{fact}{Fact}
\theoremstyle{definition}
\newtheorem{definition}{Definition}

\title{Formalising Sensor Topologies for Target Counting\footnote{This work is supported by the Engineering and Physical Sciences Research Council, under grant EP/N007565/1 (S4: Science of Sensor Systems Software).}
}

\author{
 Sven Linker
\institute{University of Liverpool\\ Department of Computer Science, Liverpool, UK}
\email{s.linker@liverpool.ac.uk}
\and 
Michele Sevegnani
\institute{University of Glasgow\\
School of Computing Science, Glasgow, UK}
\email{michele.sevegnani@glasgow.ac.uk} 
}

\def\titlerunning{Formalising Sensor Topologies for Target Counting}
\def\authorrunning{S. Linker \& M. Sevegnani}
\begin{document}
\maketitle

\tikzset{%
target counting/.style={%
  target/.style = { star, star points=5, star point ratio=2.25, draw, fill=blue!20, inner sep=+0pt,  minimum size=3mm},
  sensor/.style = {draw, circle, fill=black, minimum size = 2mm, inner sep=+0pt },
  sensor range/.style = {draw, circle, dotted, minimum size = 3cm, inner sep= +0pt, thick},  
  sensor range irregular/.style = {draw,  dotted, inner sep= +0pt, thick}  
}
}

\tikzset{
  reverseclip/.style={insert path={(-10cm,-10cm) rectangle (30cm,30cm)}}
}

\begin{abstract}
We present a formal model developed to reason about topologies created by sensor ranges.
This model is  used to formalise the topological aspects
of an existing counting algorithm to estimate the number of targets in 
the area covered by the sensors. To that end, we present
a first-order logic tailored to specify relations between parts of
the space with respect to sensor coverage. The logic
serves as a specification language for Hoare-style
proofs of correctness of the topological computations of the algorithm, which
uncovers ambiguities in their results. 
Subsequently, we
extend the formal model as a step towards improving the estimation
of the algorithm.
 Finally, we sketch how the model can
be extended to take mobile sensors and temporal aspects
into account.
\end{abstract}

\section{Introduction}
\label{sec:intro}
The \emph{target counting problem} is the computational task of
estimating the total number of observable targets within a region
using local counts performed by a set of networked sensors. With
applications ranging from agriculture and wildlife protection to
traffic control and indoor security, this problem has emerged as one
of the most important applications of present and future sensor
networks. Several different kinds of sensors may be deployed depending
on the application domain. In this work, we focus on numeric
photo-electric sensors capable of counting but not identifying targets
in their vicinity, i.e., within their sensing range. Moreover, it is
assumed that the exact position of the sensors and the geometry of
their ranges is fully known. This information is referred to as the
\emph{topology} of a sensor network. 
The complexity of this problem lies in the fact that multiple sensors
may be observing the same target if it is located within the
intersection of their overlapping sensing ranges. This may lead to
wrong estimates as duplicate observations, together with the inability
to distinguish different targets, introduce over-counting. Various
algorithms have been developed to mitigate this issue. For example,
statistical information about target distribution is used
in~\cite{stat} and topological properties of the sensor network are
exploited in~\cite{Baryshnikov2009}. The reader is referred to the
recent work of Wu et al.~\cite{Wu2014} for a complete survey.
 
In our paper, we introduce a novel logical formalisation to reason about
topologies arising from overlapping sensor ranges. We then use this logic to
analyse how sensor topology can affect counting accuracy and to identify some
hidden assumptions in SCAN~\cite{Gandhi2008}, a geometry-based counting
algorithm. We remark that the generality of our approach allows us to carry out
the same kind of analysis for different target counting algorithms. Hence, we
believe our work is an important first step towards the formal analysis of the
topological properties crucial for the correct functioning of the systems we find
in domains such as Smart Cities and IoT.

In this work, we draw from the experience of one of the authors in
defining application-specific logics such as the spatial logic
introduced in~\cite{linker2015proof} for proving safety of traffic
manoeuvres on a multi-lane highway.
While the research on modal logics to reason about spatial properties
has been on the rise in the last decades~\cite{Aiello2007}, the focus
was mostly on the general properties of the logics themselves, for
example, decidability issues and complexity analysis of satisfaction.
Furthermore, the basic entities of the semantics seem to be
ill-fitting for our purposes. 
Consider the logic \(S4\), e.g.,
which is often considered to be the prototypical logic capturing the
notion of ``nearness''. Spatial models for \(S4\) are build upon a
topology, including an operator to compute the interior (or closure)
of sets. The semantics of a formula is 
defined with respect to a
point \(x\) in the topology, where the modal operator expresses that
properties hold for all points ``close'' to \(x\). However, in our
case, we are not interested in properties of single points, or 
their immediate surroundings, but rather in the properties of whole
regions of space, for example the intersections between sensor
ranges. More recently, 
 Ciancia et al. \cite{Ciancia2014} used  general closure spaces as 
semantics. However, the general issue with this approach for
us remains: while  overlapping regions 
can be modelled, it is cumbersome to reason  about them. 
 An approach more closely related to our intentions is the
\emph{Region Connection Calculus} (RCC)~\cite{Randell1992}. In RCC,
the basic entities are convex connected regions of space, and the
formulas describe the relations between them. Still, while it is easily
expressible that two regions have a non-empty intersection, it is not
possible to refer to only this intersection.

Our article is structured as follows. The next section gives a
concise overview of the SCAN algorithm, with a particular focus on its
topological aspects. Sect.~\ref{sec:static} introduces STL, a logic to
reason about the static topologies arising from overlapping sensor
ranges in target counting scenarios. Topological temporal changes can
be reasoned upon by using the temporal extension of STL defined in
Sect.~\ref{sec:dynamic-topology}. Finally, Sect.~\ref{sec:conc} ends
with conclusion and future work.


\section{SCAN algorithm}\label{sec:scan}
SCAN is a geometry-based target counting algorithm introduced by Gandhi et
al.~\cite{Gandhi2008} that computes an estimate of the number of targets by
producing upper and lower bounds on the target population. These bounds are
proved to provide worst-case guarantees for the target count for any choice of
targets and sensor ranges. In this approach, each target is modelled as a point
and sensor ranges are assumed to be convex shapes. This allows the algorithm to
handle realistic sensing scenarios, such as those characterised by asymmetric,
irregular, and anisotropic sensor ranges. Furthermore, it is assumed each sensor
is capable of counting exactly the number of targets present in its range, i.e.,
ideal sensing,\footnote{The sketch of an algorithm supporting sensing errors is
  given in~\cite{Gandhi2008}. We ignore it here as it falls outside the scope
  of this work.} the topology is fully known, and the spatial distribution of
targets is unknown.

The key idea behind SCAN is that, in general, a constant-factor approximation of
the true target count is impossible in two-dimensional space. However, an easy
approximation within factor $\sqrt{m}$ is possible, where $m$ is the maximum
\emph{degree of overlap} among sensor ranges. The degree of overlap is defined
here as the maximum number of sensing ranges covering a point in the plane. A
concise description of the various phases of SCAN is reported next.

\subsection{Definition}
Let $S$ (sensors), $C$ (counts), and $R$ (ranges) be sets ranged over by $s_i,
c_i, r_i$, respectively.
We write $c_i$ to indicate the number of targets detected by $s_i$. Similarly,
$r_i$ denotes the range of sensor $s_i$. The SCAN algorithm consists of the
following three steps:
\begin{enumerate}
\item\label{step_1} Topology $R$ is minimised by removing all redundant sensor
ranges, i.e., sensor ranges which are fully covered by the combination of one or more
other sensor ranges. The minimised topology is indicated by $R'$,
while the corresponding sensors and their readings are sets $S'$ and $C'$,
respectively.  Intuitively, no range can be deleted from $R'$ without losing some
coverage of the underlying space.
\item\label{step_2} The maximum overlap degree ($m$) of $R'$ is computed.
\item Degree $m$ is then used to compute
  the estimated number of targets $\hat{t}$ and its corresponding bounds:
  \begin{align*}
    \hat{t}=\frac{\sum_{c_i \in C'} c_i}{\sqrt{m}}\,,&& \left[\frac{\hat{t}}{\sqrt{m}},\, \hat{t}\sqrt{m}\right]\,.
  \end{align*}
\end{enumerate}
Note that $\hat{t}$ is the geometric mean of the bounds.

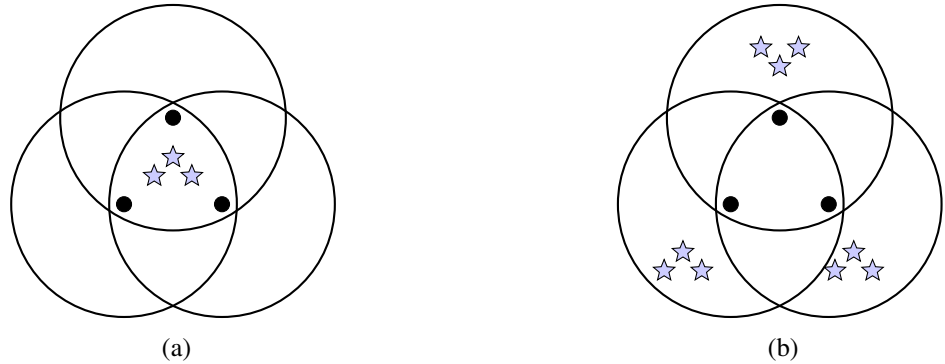
\begin{figure}
  \centering
  \begin{subfigure}[b]{0.45\textwidth}
    \centering
    \begin{tikzpicture}[target counting]

  \node[sensor] (s1) {};
  \node[sensor range, style=solid] at ($(s1)$) (r1) {};

  \node[sensor, below left = of s1,xshift =.5cm] (s2) {};
  \node[sensor range, style=solid] at ($(s2)$) (r2) {};

  \node[sensor, below right = of s1,xshift =-.5cm] (s3) {};
  \node[sensor range, style=solid] at ($(s3)$) (r3) {};

  \node[target, below = of s1, yshift =.75cm] (t1) {};
  \node[target,  below = of s1, xshift=.25cm, yshift=.5cm] (t2) {};
  \node[target, below = of s1, xshift=-.25cm, yshift=.5cm] (t3) {};

\end{tikzpicture}
    \caption{}\label{fig:scan_ex_a}
  \end{subfigure}
  \qquad%
  \begin{subfigure}[b]{0.45\textwidth}
    \centering
    \begin{tikzpicture}[target counting]

  \node[sensor] (s1) {};
  \node[sensor range, style=solid] at ($(s1)$) (r1) {};

  \node[sensor, below left = of s1,xshift =.5cm] (s2) {};
  \node[sensor range, style=solid] at ($(s2)$) (r2) {};

  \node[sensor, below right = of s1,xshift =-.5cm] (s3) {};
  \node[sensor range, style=solid] at ($(s3)$) (r3) {};

  \node[target, above = of s1, yshift =-.5cm] (t11) {};
  \node[target,  above = of s1, xshift=.25cm, yshift=-.25cm] (t21) {};
  \node[target, above = of s1, xshift=-.25cm, yshift=-.25cm] (t31) {};

  \node[target, below left = of s2, yshift =.5cm, xshift=.5cm] (t11) {};
  \node[target,  below left = of s2,  xshift=.75cm, yshift=.25cm] (t21) {};
  \node[target, below left = of s2, xshift=.25cm, yshift=.25cm] (t31) {};

  \node[target, below right = of s2, yshift =.5cm, xshift=.5cm] (t11) {};
  \node[target,  below right = of s2,  xshift=.75cm, yshift=.25cm] (t21) {};
  \node[target, below right = of s2, xshift=.25cm, yshift=.25cm] (t31) {};
  
\end{tikzpicture}
    \caption{}\label{fig:scan_ex_b}
   \end{subfigure}
  \caption{Example topologies with three~\protect\subref{fig:scan_ex_a}
    and nine~\protect\subref{fig:scan_ex_b} targets. In both examples,
    each sensor counts $3$ targets, the maximum overlap degree is $3$,
    and the estimated count computed by SCAN is $3\sqrt{3}\approx 5.2$, which
    leads to bound $\left[3,9\right]$.}\label{fig:scan_ex}
\end{figure}
The two example topologies in Fig.~\ref{fig:scan_ex} can help shed some
light on how SCAN works and why the estimation bounds are defined in this
manner. In both examples, there are three non-redundant sensors, the maximum
overlap degree is $3$ and each sensor counts three targets.  However, in
Fig.~\ref{fig:scan_ex_a}, the total number of targets is $3$, while, in
Fig.~\ref{fig:scan_ex_b}, it is $9$. In both cases, SCAN computes the estimated
number of targets $\hat{t}=3\sqrt{3}\approx 5.2$ and its bounds $\left[3,\,9\right]$.
Observe that the sensor counts do not allow to distinguish the two cases.  The
lowest possible true count occurs when all the targets are over-counted (three
times as the max overlap degree is $3$), i.e., $3$ as in
Fig.~\ref{fig:scan_ex_a}. The maximum true count arises when the readings are
not over counted, i.e., $9$ as in Fig.~\ref{fig:scan_ex_b}. These are exactly the
bounds computed by SCAN.

\paragraph{Discussion.}
The SCAN algorithm is  computationally efficient and
easy to deploy at a base station knowing the precise geometry of the
sensors’ locations and ranges.
However,
the algorithm gives acceptable estimates (i.e., tight bounds) only
when the maximum overlap degree 
is low.
This follows from the definition of estimate bounds and their width:
$s - \frac{s}{m}$, with $s=\sum_{c_i \in C'} c_i$. Another limitation
of SCAN is that the presentation given by its authors fails to include
a rigorous procedure to minimise topologies in two-dimensional
space. This is a crucial issue as we will discuss extensively in the
next section. Also, the algorithm disregards all the readings produced
by redundant sensors, raising the question if this information could
otherwise be used to obtain better estimates.



\section{Static sensor topologies}
\label{sec:static}
In this section, we will formalise the topological aspects of the 
SCAN algorithm presented in Sect.~\ref{sec:scan}. That is, we present
an algorithm to represent step~\ref{step_1}, by
removing sensors from 
 a given topology, such that all sensors are needed for
coverage. Furthermore, we present means to compute the maximum number
of overlapping sensor ranges in such a topology, to cover step~\ref{step_2}. To that end, we present
the \emph{Sensor Topology Logic} (STL) to reason about topologies created by overlapping sensor
ranges. STL is a two-sorted first-order logic with equality  and a single predicate
between  the sorts.
The models of STL are abstractions of sensor range topologies, inspired by the formalisations
of Euler diagrams \cite{Howse2002}. 
 We will use this logic to specify properties arising from two algorithms implementing
the topological computations of SCAN in two dimensions.
Furthermore, we prove these algorithms
to be correct using Hoare logic~\cite{Hoare1969}. However, we will not 
present fully formal Hoare logic proofs, but rather \emph{proof outlines}~\cite{Owicki1976}.
Note, our models are qualitative, and hence invariant under many geometric deformations.

\paragraph{Program verification using proof outlines.} 
A proof outline consists of the program to prove, annotated with assertions between
the program statements.\footnote{For more in-depth discussions, see Hoare's original contribution~\cite{Hoare1969} or Owicki and Gries' extension
to parallel programs~\cite{Owicki1976}.} These assertions are enclosed in curly braces and consist
of properties in a specification language, which is STL in our case. An annotated
 statement \(\{P\} s \{Q\}\) denotes that  \(Q\) holds after the execution of \(s\) under
the precondition \(P\). That is, if \(P\) holds initially, \(Q\) holds after executing
\(s\). If two assertions \(P\) and \(Q\) directly follow each other in a proof outline, i.e.,
\(\{P\}\{Q\}\), then we have to show that \(P\) implies \(Q\). For conditional statements \emph{if~\(P\)~then~\dots~else~\dots~end~if}, the condition
is added as an assertion to the beginning of the \emph{if} branch, while its negation is the first
annotation on the \emph{else} branch.  That is, we end up with the structure \emph{if~\(P\)~then~\(\{P\}\)~\dots~else~\(\{\lnot P\}\)~\dots~end~if}.  
 To verify a \emph{while} loop, i.e., \emph{while~\(Q\)~do~\dots~end~while}, we have to identify
the \emph{loop invariant} \(P\). If we can show that \(P\) holds before entering the loop and
after an execution of the loop, then we can also infer that \(P\) holds after leaving
the loop. Then, the conjunction of the invariant \(P\) and the loop condition \(Q\) is the first
assertion within the loop, while the conjunction \(P \land \lnot Q\) is the first assertion
after the loop. That is, we have \emph{\(\{P\}\)~while~\(Q\)~do~\(\{P \land Q\} \)~\dots~\(\{P\}\)~end~while~\(\{P \land \lnot Q\}\)}.    

\subsection{Syntax and semantics}
\label{sec:syntax-semantics}
The language of STL contains two sorts. These sorts are \emph{sensors} \(\sensors\) and \emph{zones} \(\zones\). 
The sort of sensors is countably infinite and
the sort for zones  is derived from this sort. In particular, \(\zones\) is the powerset of \(\sensors\).
However, models for the logic 
will by definition be finite. 
In addition to Boolean operators and first-order quantification, STL formulas only contain two predicates: equality
and the element relation \(\element{z}{\sense{s}}\), denoting that the zone \(z\) is contained in the 
sensor range of \(s\). It can also be read as the element relation of set theory, i.e., 
``the zone \(z\) is an element of the range of sensor  \(s\)''.

\begin{definition}[Syntax of STL]
\label{def:static-syntax}
  For a given \(n \in \N\), the syntax of the logic is given by the following definitions.
For  each  \(1 \leq i \leq n\), \(s_i\) is a \emph{sensor constant}. 
We denote the set of variables denoting zones or sensors by \(\text{Var}\). 
If \(x\) is either a sensor constant or a sensor variable, we call \(x\) a \emph{sensor term}.
Formulas are given by the following EBNF:
  \begin{align*}
    \varphi ::=  \bot \mid \theta = \theta \mid \element{z}{\sense{s}} \mid  \varphi \implies \varphi \mid \forall x \colon \varphi 
  \end{align*}
  where \(\theta\) are terms of the same sort, i.e., either zone variables or sensor terms, \(x\) is a variable, \(z\) is a zone variable and \(s\) is a sensor term.
  The other Boolean connectives (\(\land, \lor, \dots\)) and the existential quantifier are given by the usual abbreviations.
\end{definition}

 Models of STL consist of a finite set of sensors \(S\), 
a set of zones \(Z\) defining the topology, and a function \(\semSenseSing\), where
\(\semSenseSing\) associates a set of zones to each sensor, to model its range.
Formally, we have the following definition.
\begin{definition}[Static Topology Model]
\label{def:static-topology}
 A \emph{static topology model}  is a structure
\(\model = (S, Z, \semSenseSing,  \mathbf{s_1}, \dots, \mathbf{s_n} )\), where \(S\) is a set of 
sensors of sort \(\sensors\), \(Z\subseteq \powset{S}\) with \(\emptyset \not\in Z\) and for all \(\sensor{s} \in S\), there is at least
one \(z \in Z\) such that \(\sensor{s} \in z\). Furthermore, 
\(\semSenseSing \colon S \to \powset{Z}\) is a function of arity 1, 
and  \(\sensor{s_1}\) to \(\sensor{s_n}\) 
are enumerating \(S\). For all \(\sensor{s_i}\),  
 \(\semSenseSing \) is defined by \(\semSense{\sensor{s_i}} = \{z \mid \sensor{s_i} \in z \land z \in Z\}\). 
We will also write \(\sensor{s} \in \model\) for \(\sensor{s} \in S\). The set of all static topology models is denoted by \(\statmodels\).
In the following, we will often omit the distinguished elements \(\sensor{s_i}\) from models.
We denote the \emph{empty model}, i.e., the model where \(S = \emptyset\), simply by \(\emptyset\). 
The \emph{size} of a model \(\lattice = (S,Z,\semSenseSing)\), denoted by \(\size{\model}\), is the number of sensors within \(\model\), i.e.,
\(\size{\model} = \size{S}\). Furthermore, we will use the notation \(S_\model\), \(Z_\model\), and \(\semSenseSing_\model\) to refer
to the respective elements of \(\model\).
\end{definition}
Observe that each sensor constant of the syntax is associated with exactly one element of the sensor domain of a model. While
we denote the syntactic sensor constant with letters in italics, e.g., \(s\), \(t\), \(a\), \dots, we use
 bold letters for the semantic sensors \(\sensor{s}\), \(\sensor{t}\), \(\sensor{a}\), respectively.

\begin{figure}
\begin{subfigure}[b]{.49\linewidth}
\begin{center}
\begin{tikzpicture}[framed, target counting]
  \node[sensor, label=\(a\)] (s1) {};
  \node[sensor, label=\(b\), below right = of s1, yshift = -1cm] (s2) {};
  \node[sensor, label=\(c\), right = of s1, xshift=1cm] (s3) {};
  \node[sensor, label=\(d\), right = of s3, xshift = -.25cm, yshift=.25cm] (s4) {};
  
  \node[sensor range, solid] at ($(s1)$) (r1) {};
  \node[sensor range,  solid] at ($(s2)$) (r2) {};
  \node[sensor range, minimum size=3.75cm, solid] at ($(s3)$) (r3) {};
  \node[sensor range, solid, minimum size = 1.5cm] at ($(s4)$) (r4) {};

\begin{scope}
  \clip ($(s1)$) circle (1.5cm);
  \clip ($(s3)$) circle (1.875cm);
  \clip[reverseclip] ($(s2)$) circle (1.5cm);
  \fill[blue, opacity=.25] ($(s3)$) circle (3.75cm);
\end{scope}

\path[fill=green, opacity=.25] ($(s2)$) circle (1.5cm);
\end{tikzpicture}
\caption{Example of a Formal Topology}
\label{fig:example-topology}
\end{center}
\end{subfigure}
\begin{subfigure}[b]{.49\linewidth}
\begin{center}
\begin{tikzpicture}[framed, target counting]
  \node[sensor, label=\(a\)] (s1) {};
  \node[sensor, label=\(b\), right= 1cm of s1] (s2) {};
  \node[sensor, label=\(c\), left=1cm of s1] (s3) {};
  \node[sensor range irregular,  rounded corners, minimum height=3.1cm, minimum width =5.5cm] at ($(s1)$) (r1) {};

  \node[sensor range irregular, rounded corners, solid, minimum size=3cm] at ($(s2)$) (r2) {};
  \node[sensor range irregular, rounded corners, solid, minimum size=3cm] at ($(s3)$) (r3) {};
\end{tikzpicture}
\caption{Example of Ambiguous Irreducible Models}
\label{fig:example-ambiguity}
\end{center}
\end{subfigure}
\caption{Examples of Topologies}
\label{fig:examples}
\end{figure}
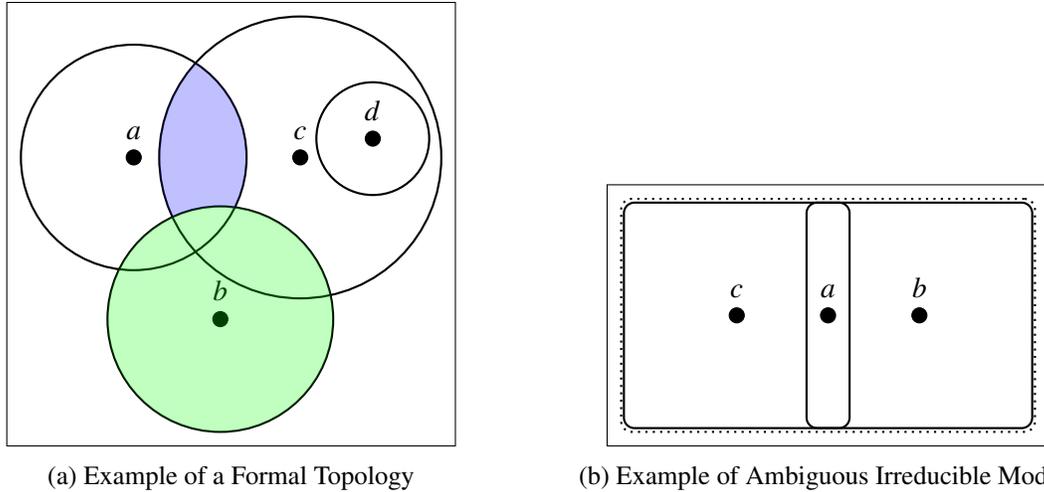

Consider the topology depicted in Fig.~\ref{fig:example-topology}, consisting of four sensors \(\sensor{a}\), \(\sensor{b}\), \(\sensor{c}\), \(\sensor{d}\) and their 
respective sensor ranges. Within its formalisation \(\model\), the intersection of \(\sensor{a}\) and \(\sensor{c}\) (coloured blue
in the figure) is formalised by a zone \(\{\sensor{a},\sensor{c}\} \in Z\). The sensor range of \(\sensor{b}\) (coloured green) consists of
all zones containing \(\sensor{b}\). It is given by the value
of the function \(\semSenseSing\), i.e., \(\semSense{\sensor{b}} = \{ \{\sensor{b}\}, \{\sensor{a},\sensor{b}\}, \{\sensor{b},\sensor{c}\},  \{\sensor{a},\sensor{b},\sensor{c}\} \}\).
Note that for a non-empty model, the definition of \(\semSenseSing\) and \(Z\) ensures that \(\semSense{\sensor{s}} \neq \emptyset\) for all \(\sensor{s} \in S\).
Furthermore, the empty model is uniquely defined as \((\emptyset, \emptyset, \emptyset)\).

\begin{definition}[Semantics]
\label{def:static-semantic}
A \emph{valuation} 
for a given model \(\model = (S,Z, \semSenseSing)\) is a function \(\val \colon \text{Var} \to S\cup Z \), 
where \(\val(s) = \sensor{s}\) and which respects
the sorts of terms. 
The semantics of a formula with respect to a model \(\model\) and a valuation
\(\val\) is given by the following inductive definition.
\begin{align*}
\model, \val & \not\models \bot && \text{for all } \lattice \text{ and } \val\\
  \model, \val & \models  \theta_1 = \theta_2 & \iff&   \val(\theta_1) = \val(\theta_2)\text{, where } \theta_1 \text{ and } \theta_2 \text{ are of the same sort}\\
\model, \val & \models \element{z}{\sense{s}} &\iff&  \val(z) \in \semSense{\val(s)} \text{, where } z \text{ is a term of sort } \zones \text{ and } s \text{ of } \sensors\\
\model, \val & \models \varphi_1 \implies \varphi_2 &\iff&  \model, \val \models \varphi_1 \text{ implies } \model, \val \models \varphi_2 \\
\model, \val & \models \forall x \colon \varphi  &\iff& \text{for all correctly sorted } \alpha \in  \model, \val\subst{x}{\alpha} \models \varphi
\end{align*}
A formula \(\varphi\) is \emph{valid in a model \(\model\)}, denoted by \(\model \models \varphi\), if \(\model, \val \models \varphi\) 
for all valuations \(\val\). Similarly, a formula \(\varphi\) is \emph{valid}, written \(\models \varphi\), if it is valid in all models.
\end{definition}

Using STL, we can express relations between different sensor ranges. 
In particular, we can mimic 
the subset relation, and 
we overload the equality sign to cover full sensor ranges as well.
\begin{align*}
  \sense{s} \leq \sense{t} &\equiv \forall x \colon \element{x}{\sense{s}} \implies \element{x}{\sense{t}}\\
\sense{s} = \sense{t} & \equiv \forall z \colon \element{z}{\sense{s}} \equivalent \element{z}{\sense{t}}
\end{align*}

Furthermore, we will employ several additional abbreviations.
The formula \(\ldisjoint{s}{t}\) expresses
that the ranges of \(s\) and \(t\) are disjoint. Observe that this relation is symmetric. Indeed, it is the
negation of \(\lcommon{s}{t}\), which denotes that the ranges of \(s\) and \(t\) have a common element.
The abbreviation \(\loverlap{s}{t}\) states that the ranges of \(s\) and \(t\) have a common element, but neither
is a subset of the other. That is, the ranges of \(s\) and \(t\) are properly overlapping. Finally, 
\(\lpsubset{s}{t}\) denotes that the sensor range of \(s\) is a proper subset of \(t\)'s sensor range.
\begin{align*}
\ldisjoint{s}{t} & \equiv \forall x \colon \element{x}{\sense{s}} \implies \lnot\element{x}{\sense{t}}\\
\lcommon{s}{t} & \equiv \exists x \colon \element{x}{\sense{s}} \land \element{x}{\sense{t}}\\
\loverlap{s}{t} & \equiv  \lcommon{s}{t} \land \lnot ( \sense{s} \leq \sense{t}) \land \lnot (\sense{t} \leq \sense{s})\\
\lpsubset{s}{t} & \equiv    (\sense{s} \leq \sense{t}) \land \lnot (\sense{t} \leq \sense{s})
\end{align*}

Consider the topology of Fig.~\ref{fig:example-topology} and its formalisation \(\model\). The sensor ranges 
in \(\model\) are
given by
\begin{align*}
  \semSense{\sensor{a}} &=\{\{\sensor{a}\}, \{\sensor{a},\sensor{b}\}, \{\sensor{a},\sensor{c}\}, \{\sensor{a},\sensor{b},\sensor{c}\} \} &
  \semSense{\sensor{b}} &= \{\{\sensor{b}\}, \{\sensor{a},\sensor{b}\}, \{\sensor{b},\sensor{c}\}, \{\sensor{a},\sensor{b},\sensor{c}\}\} \\
  \semSense{\sensor{c}} & = \{\{\sensor{c}\}, \{\sensor{a},\sensor{c}\}, \{\sensor{b},\sensor{c}\}, \{\sensor{c},\sensor{d}\}, \{\sensor{a},\sensor{b},\sensor{c}\}  \} &
  \semSense{\sensor{d}} & = \{\{\sensor{c},\sensor{d}\}\}
\end{align*}
That is, we have, for example,
\begin{align*}
  \model \models & \forall z \colon \element{z}{\sense{d}} \implies \element{z}{\sense{c}}\enspace ,&
\model \models & \exists z \colon \element{z}{\sense{c}} \land \lnot \element{z}{\sense{d}}\enspace ,\\
\model \models& \lpsubset{d}{c}\enspace ,&
\model \models& \ldisjoint{a}{d} \enspace . 
\end{align*}

\subsection{Identifying redundant sensors}
\label{sec:logic-redundant}
In this section, we present the algorithm to formalise step~\ref{step_1} of the SCAN algorithm. 
The following formula can be used to identify sensors whose sensor range is covered by a union of
other sensor ranges. That is, such a sensor can be removed from a model without impact on coverage.
\begin{align*}
\lred{s} \equiv \forall x \colon \element{x}{\sense{s}} \implies \exists y \colon  y \neq s \land \element{x}{\sense{y}}  
\end{align*}
Hence, a sensor is redundant if, and only if, all parts of space within its range are also
covered by a different sensor. Using this formula, we can give the notion of an irreducible model.
\begin{definition}[Irreducible Model]
\label{def:minimal-model}
A non-empty model \(\model\)  is 
\emph{irreducible}, if it contains no redundant sensor. That is, 
\(\model \models \forall s \colon \lnot \lred{s}\).
If \(\model = (S,Z, \semSenseSing)\) is a non-empty model and
\(\model^M = (S^M,Z^M,  \semSenseSing^M)\) is an irreducible model such that \(S^M \subseteq S\), \( Z^M =\{z^M \mid \exists z \in Z \colon z^M = z \cap S^M\}\), 
 and \(\semSenseSing^M\) is the function induced by \(S^M\) and \(Z^M\),
then we call \(\model^M\) \emph{an irreducible model of \(\model\)}.
\end{definition}
Note that irreducible models of a given model are not unique. 
Consider, for example, the topology indicated in Fig.~\ref{fig:example-ambiguity}, where the 
sensor range of \(\sensor{b}\) and \(\sensor{c}\) is denoted by solid rectangles, while the sensor range of \(\sensor{a}\) is drawn
dotted. This topology is formalised by
 \(\model = (\{\sensor{a},\sensor{b},\sensor{c}\}, Z,  \semSenseSing)\), where
\(Z = \{\{\sensor{a},\sensor{b}\}, \{\sensor{a},\sensor{b},\sensor{c}\}, \{\sensor{a},\sensor{c}\}\}\),  that is, 
 \(\semSense{\sensor{a}} = \semSense{\sensor{b}} \union \semSense{\sensor{c}}\) 
and \(\semSense{\sensor{b}} \neq \semSense{\sensor{c}}\). Then both \(\model_1 = (\{\sensor{a}\}, \{\{\sensor{a}\}\}, \semSenseSing^1)  \)
and \(\model_2 = (\{\sensor{b},\sensor{c}\}, \{\{\sensor{b}\}, \{\sensor{b},\sensor{c}\}, \{\sensor{c}\}\}, \semSenseSing^2)  \) are irreducible models of \(\model\).

\begin{definition}[Reduction]
\label{def:reduction}
Let \(\model = (S, Z,  \semSenseSing)\) be a model, and \(S^\prime \subseteq S\) a proper subset
of \(S\). The \emph{reduction of \(\model\) by \(S^\prime\)} is given
by \(\model \setminus S^\prime = (S^R, Z^R, \semSenseSing^R)\), where 
\begin{enumerate}
\item \(S^R = S \setminus S^\prime\),
\item \(Z^R = \{ z^R \mid z \in Z \land z^R = z \setminus S^\prime\}\), and
\item \(\semSenseSing^R\) is the function induced by \(S^R\) and \(Z^R\) according to Def.~\ref{def:static-topology}.
\end{enumerate}
\end{definition}

\begin{figure}
  \centering
\begin{subfigure}[b]{.49\textwidth}
  \begin{algorithmic}[1]
    \Function{reduce}{\model}
    \While{\(\model \models \exists t \colon \lred{t}\)}
    \State \(s \gets \chs{\{\sensor{t} \mid \model \models  \lred{t}\}}\)
    \State \(\model \gets \model \setminus \{s\}\)
    \EndWhile
    \State \Return \(\model\)
    \EndFunction
  \end{algorithmic}
\caption{Reduce a given model to an irreducible model}
\label{alg:reduce}
\end{subfigure}
\begin{subfigure}[b]{.49\textwidth}
  \begin{algorithmic}[1]
    \Function{maximum}{\model}
    \State \(m \gets  \size{\model} \)
    \While{\(m > 1\)}
    \If{\(\model \models O^m\)}
    \State \Return \(m\)
    \Else
    \State \(m \gets m-1\)
    \EndIf
    \EndWhile
    \State \Return \(m\)
    \EndFunction
 \end{algorithmic}
\caption{Compute maximum number of overlaps}
\label{alg:non-redundant-maximum-overlaps}
\end{subfigure}
  \caption{Topological Computations of SCAN}
  \label{fig:scan-topo-alg}
\end{figure}

We can prove that every non-empty sensor topology \(\model\) contains an irreducible model. This by no means proves that this 
reduction  is \emph{unique}, since a suitable counterexample is already given in Fig.~\ref{fig:example-ambiguity}. 
However,
this lemma has to be carefully stated, since sensors that are redundant in \(\model\), will not be redundant in
 irreducible models of \(\model\). Hence, we can only state the existence of a  subset of sensors that has to be removed to
yield an irreducible model of \(\model\).  
\begin{lemma}
Every non-empty model contains a non-empty irreducible model. 
\begin{align*}
  \model \neq \emptyset \iff \exists R \subseteq \{\sensor{c} \mid \model \models \lred{c}\} \colon \model \setminus R  \neq \emptyset \text{ and } \model \setminus R \models \forall s \colon \lnot \lred{s}
\end{align*}
\end{lemma}
\begin{proof}
  The direction from right to left is immediate, since \(\model\) contains at least as many sensors as \(\model \setminus R\) for all
possible sets \(R\). Now let \(\model \neq \emptyset\) and proceed by induction on the number \(n\) of redundant sensors in \(\model\).
That is \(n = | \{ \sensor{c} \mid \model \models \lred{c}\}|\). If \(n=0\), then \(\model\) is irreducible and we 
can choose \(R=\emptyset\) to finish the base case. For the induction step, let the lemma be true for all \(n' \leq n\) and let
\(|\{\sensor{c} \mid \model \models \lred{c}\}| = n+1\). Choose one \(\sensor{d} \in  \{\sensor{c} \mid \model \models \lred{c}\}\).
Then \(|\{\sensor{c} \mid \model \setminus\{\sensor{d}\} \models \lred{c}\}| < |\{\sensor{c} \mid \model \models \lred{c}\}|\).
Hence, by the induction hypothesis, there is an \(R \subseteq \{\sensor{c} \mid \model \setminus\{\sensor{d}\} \models \lred{c}\}\) such that
 \( (\model \setminus\{\sensor{d}\})\setminus R \neq \emptyset \) and \( (\model \setminus\{\sensor{d}\}) )\setminus R \models \forall s \colon \lnot \lred{s}\).
Now let \(R^\prime = R \cup \{\sensor{d}\}\). Then we have \(\model \setminus R^\prime \neq \emptyset \) and \(\model \setminus R^\prime \models \forall s \colon \lnot \lred{s}\) and
we are done.
\end{proof}

Figure~\ref{alg:reduce} shows the algorithm allowing us to reduce a given model to an irreducible model covering the same space, thus representing
step~\ref{step_1} of SCAN. The irreducible
model returned by the algorithm is not unique, since the algorithm non-deterministically chooses a 
redundant sensor to remove from the model (denoted by the function \(\chsSing\)) . Since the algorithm removes one sensor at a time, we have
to ensure that the reduction of a model by a redundant sensor does again yield a non-empty model.
\begin{lemma}
\label{lem:redundant}
Let \(\sensor{d}\) be a redundant sensor in \(\model\), i.e., \(\model \models \lred{d}\). Then
\(  \model \neq \emptyset \Rightarrow \model \setminus \{\sensor{d}\} \neq \emptyset\, . \)
\end{lemma}
\begin{proof}
Let  \(\model \neq \emptyset \) and \(\model \models \lred{d}\). Hence, there is at least one sensor \(\sensor{s}\)
different from \(\sensor{d}\) such that \(\model \models d \neq s \land \exists x \colon \element{x}{ \sense{d}} \land \element{x}{\sense{s}}\), since
\(\semSense{\sensor{s}} \neq \emptyset\) for all \(\sensor{s}\). Due to the existence of \(\sensor{s}\), we get \(\model \setminus \{\sensor{d}\} \neq \emptyset\).  
\end{proof}
 
With this lemma, we can provide a proof outline for the reduction algorithm, as shown in Fig.~\ref{prf:alg-reduce}. Note 
that the precondition of \(\model\) being non-empty allows us to apply Lemma~\ref{lem:redundant} while choosing a redundant
sensor to remove. This also shows that the algorithm terminates, since otherwise it would have to remove redundant 
sensors infinitely often, which contradicts the finiteness of the models.

\begin{figure}[t]
\begin{subfigure}[b]{.62 \linewidth}
  \begin{algorithmic}[1]
    \assert[Pre:]{\model \neq \emptyset}
    \Function{reduce}{\model}
    \assert{\model  \neq \emptyset}
    \While{\(\model \models \exists t \colon \lred{t}\)}
    \assert{\model \models  \exists t \colon \lred{t)} \land \model  \neq \emptyset}
    \assert{\model \setminus \{\chs{\{\sensor{t} \mid \model \models \lred{t}\}}\} \neq \emptyset\ (\text{Lemma}~\ref{lem:redundant})}
    \State \(s \gets \chs{\{\sensor{t} \mid \model \models  \lred{t}\}}\)
    \assert{\model \setminus \{s\} \neq \emptyset}
    \State \(\model \gets \model \setminus \{s\}\)
    \assert{\model \neq \emptyset}
    \EndWhile
    \assert{\lnot (\model \models  \exists t \colon \lred{t})  \land \model  \neq \emptyset}
    \assert{\model \models  \forall t \colon \lnot \lred{t}  \land \model  \neq \emptyset}
    \State \Return \(\model\)
    \EndFunction
  \end{algorithmic}
\caption{Proof Outline for Alg.~\ref{alg:reduce}}
\label{prf:alg-reduce}
\end{subfigure}
\begin{subfigure}[b]{.37\linewidth}
  \begin{algorithmic}[1]
 \assert[Pre: ]{\model \neq \emptyset}
    \Function{maximum}{\model}
    \assert{\size{\model}\geq r_\model \land r_\model \geq 1}
    \State \(m \gets  \size{\model} \)
    \assert{m \geq r_\model}
    \While{\(m > 1\)}
    \assert{m>1 \land m\geq r_\model}
    \If{\(\model \models O^m\)}
    \assert{m\geq r_\model \land \model \models O^m}
    \assert{m = r_\model\ (\text{Fact}~\ref{fact:one-overlap})}
    \State \Return \(m\)
    \Else
    \assert{m\geq r_\model \land \model \not\models O^m}
    \assert{ m-1 \geq r_\model\ (\text{Fact}~\ref{fact:one-overlap})}
    \State \(m \gets m-1\)
    \assert{m \geq r_\model}
    \EndIf
    \EndWhile
    \assert{m \leq 1 \land m\geq r_\model}
    \assert{m = r_\model\ (\text{Fact}~\ref{fact:one-overlap}, \text{Lemma}~\ref{lem:minimum-overlap})}
    \State \Return \(m\)
    \EndFunction
  \end{algorithmic}
  \caption{Proof Outline for Alg.~\ref{alg:non-redundant-maximum-overlaps}}
  \label{prf:alg-non-redundant-maximum-overlaps}
\end{subfigure}
\caption{Proof Outlines for Topological Computations of SCAN}
\label{prf:scan}
\end{figure}
\subsection{Compute maximum amount of overlap}
\label{sec:logic-maximum-overlap}
In the following section, we show how to formalise step~\ref{step_2} of SCAN.
To that end, we define
the following formula for each \(m\) (bounded by the number of sensors in our model).
\begin{align*}
  O^m &\equiv \exists x,  y_1, \dots, y_m \colon \bigwedge_{i=1}^m \element{x}{\sense{y_i}} \land \bigwedge_{i \neq j} y_i \neq y_j
\end{align*}
That is, \(O^m\) is true, if, and only if there is a zone in the model which is covered by \(m\) distinct sensors. Observe that the 
sorting of terms  (in particular, the application of \(\senseSing\)) ensures that all the \(y_i\) are indeed sensors and not 
arbitrary elements of the space represented by \(Z\).  
Then, we can show that every topology model contains  a zone covered by at least one sensor.
\begin{lemma}
\label{lem:minimum-overlap}
  Let \(\model \neq \emptyset\). Then \(\model \models O^1\). 
\end{lemma}
\begin{proof}
If \(\model \neq \emptyset\), we know that there is at least
one \(\sensor{s} \in S\) and consequently one \(z \in \semSense{\sensor{s}}\). 
Observe that \(O^1 \equiv \exists x, y \colon \element{x}{\sense{y}}\). Now, let \(\val\) an
arbitrary valuation. During the evaluation of the existential quantifier, we can choose
the new value for \(x\) to be \(z\) and for \(y\) to be \(\sensor{s}\). Hence
\(\model, \val \models O^1\).
 \end{proof}

Figure~\ref{alg:non-redundant-maximum-overlaps} presents the formal algorithm
of the computation of the maximal number of overlapping sensors. To prove it to be correct, we need
some minor definitions. We let
\(r_\model\) be the maximum of the number of non-redundant overlapping sensors in the model
\(\model\), i.e.,
\begin{align*}
  r_\model = \max\{m \mid \model \models O^m\}\enspace .
\end{align*}
We want to show that our algorithm always returns \(r_\model\). 
For a non-empty model \(\model = (S,Z,\semSenseSing)\), 
we have \(r_\model \geq 1\) by Lemma~\ref{lem:minimum-overlap}. Furthermore, the definition of \(r_\model\) establishes the 
following fact.
\begin{fact}
\label{fact:one-overlap}
The formula \(O^m\) is valid in a static topology model \(\model\), if and only if \(m\) is at most the maximum amount of overlaps within \(\model\).
That is,  \(\model \models O^m \iff m \leq r_\model\).
\end{fact}
With these notions at hand, we are able to prove  correctness  of the algorithm in Fig.~\ref{alg:non-redundant-maximum-overlaps} (see Fig.~\ref{prf:alg-non-redundant-maximum-overlaps}). 
We do not 
refer to topological properties in the proof outline, hence the specification language
is just  arithmetic. 
The loop invariant is that the variable \(m\) is greater or 
equal to the maximal amount of overlapping sensor ranges, i.e. \(m \geq r_\model\).
Observe that the algorithm terminates due to the finiteness of the models.

\subsection{Non-destructive target estimation}
\label{sec:dyn-target-estimation}
In the previous sections, we formalised  the topological aspects of the SCAN algorithm.
However, the computation we presented is destructive,
i.e., we remove sensors from a given model until we get a irreducible model. 
Hence, we removed information that could be used 
to tighten the estimation  bounds computed by the algorithm. In this section, we show
how this construction may be altered to preserve this information. To that
end, instead of removing sensors, we will mark the sensors to be either
necessary or unnecessary for the computation. This slightly complicates the algorithms,
since we cannot simply identify the necessary sensors by finding zones that 
are only contained within a single sensor range, but we instead have  to disregard
the sensors already marked as unnecessary.

Hence, we extend the syntax and semantics with  predicates \(\necessarySymb\) and \(\unnecessarySymb\) to express, whether a sensor
is considered (un-)necessary to cover the given space. To that end, we extend a model over the set of sensors \(S\) 
with two new sets \(\necessarySem \subseteq S\) and \(\unnecessarySem \subseteq S\), where \(\necessarySem \cap \unnecessarySem = \emptyset\). 
Syntactically, both predicates 
are unary, and their semantics is given by
\begin{align*}
\model, \val & \models \necessary{s} &\iff& \val(\sensor{s}) \in \necessarySem \,, \\
\model, \val & \models \unnecessary{s} &\iff& \val(\sensor{s}) \in \unnecessarySem \,. 
\end{align*}
The predicates \(\necessarySymb\) and \(\unnecessarySymb\) allow us during the algorithms to categorise sensors used
for computing the estimation of targets. 
 Consider the example given in Sect.~\ref{sec:logic-redundant} 
for irreducible models.
 If we mark \(\sensor{a}\) as necessary, but do not remove \(\sensor{b}\) and \(\sensor{c}\), \(\sensor{a}\)
would still satisfy the redundancy formula. We rectify this situation by relaxing our notion of
redundancy to capture \emph{possible redundancy}. 
\begin{align*}
  \lnred{s} &\equiv \forall x \colon \element{x}{\sense{s}} \implies \exists y\colon y \neq s \land \lnot \unnecessary{y} \land \element{x}{\sense{y}}\\  
\lcovered & \equiv \forall z \exists s \colon \lnot \unnecessary{s} \land \element{z}{\sense{s}}
\end{align*} 
In contrast to \(\lred{s}\), the formula \(\lnred{s}\) requires that the covering sensors must not be included in the set \(\unnecessarySem\), i.e., they may not be deemed
unnecessary for full coverage. 
Of course, we have to modify the definition of irreducible models as well, since we will keep redundant sensors in
the models.
\begin{definition}[Irreducible Model (2)]
  Let \(\model = (S,Z, \necessarySem, \unnecessarySem, \semSenseSing)\) be a non-empty static topology 
model. We call \(\model\) an \emph{irreducible model}, if and only if it satisfies the following formula
\begin{align*}
(  \forall x \exists y \colon \element{x}{\sense{y}} \land \necessary{y}) \land (\forall s \colon \necessary{s} \lor \unnecessary{s})  \land (\forall s \colon  \lnred{s} \implies \unnecessary{s}) .
\end{align*}
That is, every zone in \(\model\) is covered by at least one sensor marked as necessary, every sensor is either necessary or unnecessary, and 
every possibly redundant sensor is indeed marked as unnecessary.
\end{definition}

Figure~\ref{alg:reduce-non-dest} shows the algorithm for non-destructive reduction. It is similar
to the destructive algorithm, but instead of completely removing the sensors deemed unnecessary, it
just marks them appropriately. Finally, all remaining sensors are marked as necessary. 
While we do not show a formal proof of termination, observe that in each step, a new sensor is marked as unnecessary. Since
the model contains only finitely many sensors, the algorithm has to terminate.

\begin{figure}
  \begin{algorithmic}[1]
    \Function{reduce}{\model}
    \While{\(\model \models \exists s \colon \lnred{s}\)}
    \State{\( s \gets \chs{\{\sensor{s} \mid \model \models \ \lnred{s}\}}\)}
    \State{\(\model \gets (S_\model, Z_\model,N_\model,U_\model \cup \{s\},\semSenseSing_\model)  \)}
    \EndWhile
    \State{\(\model \gets (S_\model, Z_\model,S_\model \setminus U_\model,U_\model,\semSenseSing_\model)  \)}
    \State \Return \(\model\)
    \EndFunction
  \end{algorithmic}
\caption{Non-Destructive Reduction}
\label{alg:reduce-non-dest}
\end{figure}

The corresponding
proof outline is presented in Fig.~\ref{prf:reduce-non-dest}. The loop invariant is the formula \(\lcovered\), i.e.,
at each part of the loop, all zones in the model are covered by at least one sensor not yet marked as unnecessary. 
Most of the proof steps are direct application
of the Hoare proof rules. However, there are three steps that are not immediate,
First, we have to prove that the input of the algorithm, a model
without any marked sensors, satisfies the loop invariant. This is shown in Lemma~\ref{lem:no-unnec-coverage}.
\begin{lemma}
\label{lem:no-unnec-coverage}
Let \(\model\) be a non-empty model with \(\necessarySem_\model = \unnecessarySem_\model = \emptyset\). Then
\(\model \models \lcovered\).
\end{lemma}
\begin{proof}
Since \(\unnecessarySem_\model = \emptyset\), no sensor is marked as unnecessary. Furthermore, due to our
general condition on topology models, each zone is covered by at least one sensor. All in all, each zone
is covered by at least one sensor, which is not marked as unnecessary.
\end{proof}

\begin{figure}
  \begin{algorithmic}[1]
    \assert[Pre:]{\model \neq \emptyset \land \necessarySem_\model = \unnecessarySem_\model= \emptyset}
    \Function{reduce}{\model}
    \assert{\model \models \lcovered\ (\text{Lemma}~\ref{lem:no-unnec-coverage})}
    \While{\(\model \models \exists s \colon \lnot \unnecessary{s} \land  \lnred{s}\)}
    \assert{ \model \models \lcovered \text{ and } \model \models \exists s \colon \lnot \unnecessary{s} \land \lnred{s}} \label{lst:inv}
    \assert{(S_\model, Z_\model,N_\model,U_\model  \cup \{\chs{\{\sensor{s}\! \mid\! \model \models \lnot\unnecessary{s}\land \lnred{s}\}},\semSenseSing_\model) \models \lcovered\}\ (\text{Lem.}~\ref{lem:rem-pos-redundant})} \label{lst:choose}
    \State{\( s \gets \chs{\{\sensor{s} \mid \model \models \lnot\unnecessary{s}\land \lnred{s}\}}\)}
    \assert{(S_\model, Z_\model,N_\model,U_\model \cup \{s\},\semSenseSing_\model) \models \lcovered }
    \State{\(\model \gets (S_\model, Z_\model,N_\model,U_\model \cup \{s\},\semSenseSing_\model)  \)}
    \assert{\model \models \lcovered}
    \EndWhile
    \assert{  \model \models \lcovered \text{ and } \model \models \forall s \colon \unnecessary{s} \lor  \lnot \lnred{s}}\label{lst:after-loop}
    \assert{(S_\model, Z_\model,S_\model \setminus U_\model,U_\model,\semSenseSing_\model) \text{ is irreducible. } (\text{Lemma}~\ref{lem:irreducible-const})} \label{lst:irred}
    \State{\(\model \gets (S_\model, Z_\model,S_\model \setminus U_\model,U_\model,\semSenseSing_\model)  \)}
    \assert{\model \text{ is irreducible.}}
    \State \Return \(\model\)
    \EndFunction
  \end{algorithmic}
\caption{Proof Outline for Alg.~\ref{alg:reduce-non-dest}.}
\label{prf:reduce-non-dest}
\end{figure}

The next lemma is used to prove the step from line~\ref{lst:inv} to line~\ref{lst:choose}. That is, if
the loop condition and the invariant hold for a model, then we can choose a possibly redundant sensor \(s\) that
has not yet been marked, mark it as unnecessary, and the resulting model still satisfies the invariant \(\lcovered\).  
\begin{lemma}
\label{lem:rem-pos-redundant}
Let \(\model\) be a model with \(\model \models \forall z \exists s \colon \lnot\unnecessary{s} \land \element{z}{\sense{s}}\) and \(s\) be a sensor constant such that
\(\model \models \lnot \unnecessary{s} \land \lnred{s}\). Furthermore, let \(\model^\prime =  (S_\model, Z_\model,N_\model,U_\model \cup \{\sensor{s}\},\semSenseSing_\model)\).
Then \(\model^\prime \models   \forall z \exists s \colon \lnot\unnecessary{s} \land \element{z}{\sense{s}}\).
\end{lemma}
\begin{proof}
Observe that the zones and their containment relations in the ranges does not change between \(\model\) and \(\model^\prime\).
For all \(z\) with \(\model^\prime \models \lnot \element{z}{\sense{s}}\), we have the result immediately, by the condition on \(\model\). 
So let \(z\) be a zone with \(\model^\prime \models \element{z}{\sense{s}}\). By \(\model \models \lnred{s}\), we know that there 
is a sensor \(t\) different from \(s\) with \(\model \models \lnot \unnecessary{t} \land \element{z}{\sense{t}}\). Since in 
\(\model^\prime\) we still have \(\model^\prime \models \lnot \unnecessary{t}\), we get \(\model^\prime \models \exists t \colon s \neq t \land \lnot \unnecessary{t} \land \element{z}{\sense{s}}\).
Since \(z\) was arbitrary, the lemma follows.
\end{proof}

The last lemma of this section refers to the final steps in the algorithm, i.e., the relation
between line~\ref{lst:after-loop} and \ref{lst:irred}. It shows that if we have a model that
satisfies the loop invariant, and where all sensors are either marked as unnecessary or
are not possibly redundant, then the model created by marking all remaining sensors
as necessary is irreducible. That is, this lemma states that the model
returned by the algorithm is always irreducible.
\begin{lemma}
\label{lem:irreducible-const}
   Let \(\model\) be model such that \(\model \models \lcovered\) and \(\model \models \forall s \colon \unnecessary{s} \lor  \lnot \lnred{s}\).
Furthermore, let \(\model^\prime =   (S_\model, Z_\model,S_\model \setminus U_\model,U_\model ,\semSenseSing_\model)\).
Then \(\model^\prime\) is irreducible, i.e., 
\begin{align}
\model^\prime \models &  \forall x \exists y \colon \element{x}{\sense{y}} \land \necessary{y} \label{eq:nec-coverage}\\
 \model^\prime \models& \forall s \colon \necessary{s} \lor \unnecessary{s} \label{eq:nec-unnec}\\
\model^\prime \models& \forall s \colon \lnred{s} \implies \unnecessary{s}  \label{eq:unnec-posred}.
\end{align}
\end{lemma}
\begin{proof}
Property~(\ref{eq:nec-unnec}) holds by construction of \(\model^\prime\). Furthermore, we get property~(\ref{eq:unnec-posred}) by
the assumption on \(\model\), and the fact that \(\model^\prime\) contains the same set of unnecessary sensors as \(\model\).
Finally, property~(\ref{eq:nec-coverage}) holds, since \(\model \models \lcovered\), and for all sensors \(s\) with \(\model \models \lnot \unnecessary{s}\),
we have \(\model^\prime \models \necessary{s}\).  
\end{proof}

For the counterpart to the algorithm shown in Fig.~\ref{alg:non-redundant-maximum-overlaps}, we can amend the formulas to only use sensors deemed necessary. The algorithm itself
does not need to be changed.
\begin{align*}
  O^m &\equiv \exists x,  y_1, \dots, y_m \colon \bigwedge_{i=1}^m \element{x}{\sense{y_i}} \land \bigwedge_{i \neq j} y_i \neq y_j \land \bigwedge_{i=1}^m \necessary{y_i}
\end{align*}



\section{Dynamical topology models}
\label{sec:dynamic-topology}
The models defined in Sect.~\ref{sec:static} define  static topologies of sensor nodes. However,
in reality,   
 topological  relations of the sensor ranges may change over time, for example 
 due to environmental reasons like rain or fog, or moving obstacles blocking the sensors, or
because the sensors are 
 mobile, e.g., attached to drones or autonomous robots.
 This may result in previously unnecessary sensors becoming necessary  and vice versa. 
In this section, we present how we can extend the formal models to reason about temporal changes in
the topologies, by temporalising our static logic with respect to \emph{full computation tree logic} (\(\text{CTL}^\ast\)) \cite{Emerson1986}. 
Since our models are qualitative in nature, we only capture purely qualitative changes
within the topology. Egenhofer and Al-Taha~\cite{Egenhofer1992} identified the possible changes in qualitative relationships for
spatial regions, i.e., regions that are non-empty, convex connected subsets of \(\mathbb{R}^2\). This matches our view
of the sensor ranges. The full set of relations between regions is shown in Tab.~\ref{tab:spat-relations}.
 
\begin{table}[h]
  \centering
  \begin{tabular}{|m{2cm}|m{2cm}||m{2cm}|m{2cm}|}
\hline
   disjoint &
    \begin{tikzpicture}
      \path (0,.5) -- (0,-.5);
      \draw[pattern=north west lines] (.5,0) circle (.25cm);
      \draw[pattern=north east lines] (1.25, 0) circle (.25cm);
      \node at (.5, .4) (A) {\(A\)};
      \node at (1.25, .4) (B) {\(B\)};

    \end{tikzpicture}
&
   contains &
    \begin{tikzpicture}
      \path (0,.5) -- (0,-.5);
      \draw[pattern=north west lines] (1,0) circle (.4cm);
      \draw[pattern=north east lines] (1, 0) circle (.25cm);
    \end{tikzpicture}\\ 
\hline 
   inside & 
    \begin{tikzpicture}
      \path (0,.5) -- (0,-.5);
      \draw[pattern=north west lines] (1,0) circle (.25cm);
      \draw[pattern=north east lines] (1, 0) circle (.4cm);
    \end{tikzpicture}&
   equal &
           \begin{tikzpicture}
      \path (0,.5) -- (0,-.5);
             \draw[pattern=north west lines] (1,0) circle (.25cm);
             \draw[pattern=north east lines] (1, 0) circle (.25cm);
           \end{tikzpicture} \\
\hline 
   meet & 
           \begin{tikzpicture}
      \path (0,.5) -- (0,-.5);
             \draw[pattern=north west lines] (.75,0) circle (.25cm);
             \draw[pattern=north east lines] (1.25, 0) circle (.25cm);
           \end{tikzpicture} &
   covers & 
           \begin{tikzpicture}
      \path (0,.5) -- (0,-.5);
             \draw[pattern=north west lines] (1,0) circle (.35cm);
             \draw[pattern=north east lines] (1.105, 0) circle (.25cm);
           \end{tikzpicture} \\
\hline 
   coveredBy & 
           \begin{tikzpicture}
      \path (0,.5) -- (0,-.5);
             \draw[pattern=north west lines] (1.105,0) circle (.25cm);
             \draw[pattern=north east lines] (1, 0) circle (.35cm);
           \end{tikzpicture} &
   overlap & 
           \begin{tikzpicture}
      \path (0,.5) -- (0,-.5);
             \draw[pattern=north west lines] (.875,0) circle (.25cm);
             \draw[pattern=north east lines] (1.125, 0) circle (.25cm);
           \end{tikzpicture} \\
\hline
  \end{tabular}
  \caption{The eight possible relations between spatial regions \(A\) and \(B\) \cite{Egenhofer1992}}
  \label{tab:spat-relations}
\end{table}

However, in contrast to their work, we have no means of identifying the boundary of a spatial region. 
Furthermore, we assume that targets are only countable, if they are fully within the sensor range,
and not only on the boundary. As a consequence, the differences between several of these relations are
not of interest to us. For example, the difference between \emph{meet} and \emph{disjoint} is only whether 
the boundaries intersect. Similarly, we cannot distinguish between \emph{covers} and \emph{contains}, as well
as \emph{coveredBy} and \emph{inside}. From each of  these three pairs, we keep the less specific relation. Hence
we end up with the five relations: \emph{disjoint}, \emph{equal}, \emph{overlap}, \emph{inside}, and \emph{contains}.  
We follow Egenhofer and Al-Taha's approach of identifying the pairs of relations with the least topological distance,
while ignoring
the relations concerning the boundaries. 
This results 
in the graph of possible qualitative changes for our model as shown in Fig.~\ref{fig:qualitative-changes-relations}. 
If two regions move relative to each other, the relation between them can only change as indicated by
the graph. For example, if two initially disjoint regions move, the relation between them cannot
immediately change such that one is a subregion of the other, but they have to partially 
overlap first.   

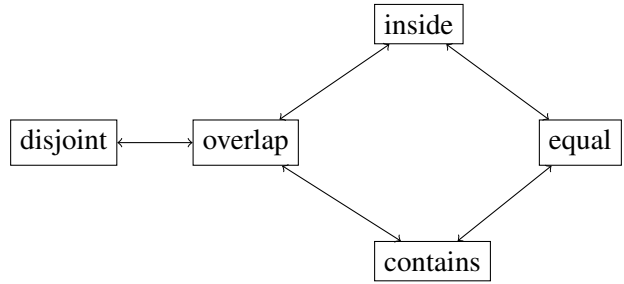
\begin{figure}[h]
  \centering
  \begin{tikzpicture}
    \node[draw] (disj) {disjoint};
    \node[draw, right=of disj] (overlap) {overlap};
    \node[draw, below right  = of overlap] (contains) {contains};
    \node[draw, above right  = of overlap] (inside) {inside};
    \node[draw, below right  = of inside] (equals) {equal};

    \draw[<->] (disj) to (overlap);
    \draw[<->] (overlap) to (contains);
    \draw[<->] (overlap) to (inside);

    \draw[<->] (equals) to (contains);
    \draw[<->] (equals) to (inside);

  \end{tikzpicture}
  \caption{Possible qualitative changes of two sensor ranges}
  \label{fig:qualitative-changes-relations}
\end{figure}

\subsection{Syntax and semantics}
\label{sec:dyn-syntax-semantics}
We extend the notion of a model along
the lines of Finger and Gabbay \cite{Finger1992}. That is, we 
use a standard possible-world semantics, but associate a static model as defined in the previous section 
to each world. The worlds are connected by directed transitions, to form a model of \emph{branching
time}. Hence, at each world, several futures are possible. While our definition of models 
does not respect the relational changes identified in in Fig.~\ref{fig:qualitative-changes-relations}, we will  
 restrict
the connections between worlds to adhere to the qualitative changes given by requiring certain
axioms to hold.

We only define the dynamic models with respect to the topologies of the sensors. That is,
we  associate a point in time with a class of models, which all possess the same topology,
but  where different sensors are marked as necessary or unnecessary.
Formally, we say that two models \(\model_1= (S_1,Z_1, N_1,U_1, \semSenseSing_1)\) and \(\model_2 =(S_2,Z_2, N_2,U_2, \semSenseSing_2)\) 
are \emph{topologically equivalent}, written
\(\model_1 \sim \model_2\), if and only if \(S_1 = S_2\), \(Z_1 = Z_2\), and \(\forall \sensor{s} \colon \semSenseSing_1(\sensor{s}) = \semSenseSing_2(\sensor{s})\).
We denote the equivalence class of a model \(\model\) by \([\model]\) and the set of these equivalence classes by \(\statmodels_{/\sim}\).

\begin{definition}[Dynamic Topology Model]
\label{def:dynmodel}
 A \emph{dynamic topology model} is a structure \(\dynmodel = (W, \trans, f)\), where \(W\) is a non-empty
set of \emph{worlds}, \(\trans \subseteq W \times W\) is the \emph{accessibility relation} and
\(f \colon W \to \statmodels_{/\sim}\) is a function mapping each world to an equivalence class of static topology models. Furthermore,
we require that the  members of the equivalence classes in the range of \(f\) are based on the same domain  of
sensors. 
\end{definition}

\begin{definition}[Path]
Let \(\dynmodel\) be a model according to Def.~\ref{def:dynmodel}. A \emph{path} \(\pi\) through \(\dynmodel\) is 
an infinite sequence of worlds \(w_0,w_1, \dots\), such that \(w_i \in \dynmodel\)  and
\(w_i \trans w_{i+1}\)  for all \( i \in \N\). Given a path \(\pi = w_0,w_1, \dots \), we will also denote
\(w_i\) by \(\pi^i\). We write the \emph{set of all paths in \dynmodel } as \(\paths{\dynmodel}{}\).  
Furthermore, the set of all paths in \(\dynmodel\) starting at \(w\) is denoted by
\(\paths{\dynmodel}{w} = \{\pi \mid \pi \in \paths{\dynmodel}{} \land \pi^0 = w\}\). 
\end{definition}

The syntax of \emph{Dynamic STL} is based on \(\text{CTL}^\ast\), but instead of propositional atoms,  we 
use static topology formulas as defined in the previous section. 

\begin{definition}[Syntax of Dynamic STL]
\label{def:dynamic-syntax}
The syntax of Dynamic STL is given by the following EBNF.
  \begin{align*}
    \chi & ::= \varphi_s \mid \lnot \chi \mid \chi \land \chi \mid \allpaths \psi\\
    \psi & ::= \chi \mid \lnot \psi \mid \psi \land \psi \mid \lnext \chi \mid \chi \until \chi 
  \end{align*}
where \(\varphi_s\) is a static topology formula, 
which does not use necessity  predicates. We call formulas of the form \(\chi\) \emph{state formulas} and
formulas of type \(\psi\) \emph{path formulas}. 
\end{definition}
The other operators, i.e., the Boolean operators like disjunction, implication, the temporal operators \(\lglobally\) and \(\lfinally\)
as well as the remaining path operator \(\somepath\) can be defined as abbreviations as usual \cite{Emerson1986}.
Observe that the semantics of purely topological formulas  is equivalent for all members of a topological equivalence
class. 
Hence we silently lift satisfaction of these formulas to topological equivalence classes 

\begin{definition}[Semantics]
\label{def:dynamic-semantics}
Let \(\dynmodel\) be a dynamic topology model and \(w \in \dynmodel\). The semantics of a formula 
with respect to \(\dynmodel\) and a valuation \(\val\) is given by the following inductive definition.
The Boolean operators are defined as usual, and hence we omit them. 
\begin{align*}
\dynmodel, w, \val & \models \varphi_s &\iff& f(w), \val \models \varphi_s \\
\dynmodel, w, \val & \models \allpaths \psi  &\iff& \forall \pi \in  \paths{\dynmodel}{w} \colon \dynmodel, \pi, \val \models \psi \\
\dynmodel, \pi,  \val & \models \lnext \chi  &\iff& \dynmodel, \pi^1,\val \models \chi\\
\dynmodel, \pi, \val & \models \chi_1 \until \chi_2 & \iff & \exists k \in \N \colon \dynmodel, \pi^k, \val \models \chi_2 \text{ and } \forall i < k \colon \dynmodel, \pi^i, \val \models \chi_1   
\end{align*}
\end{definition}

\subsection{Axiomatising semantical conditions}
\label{sec:dyn-axiomatising}
Recall the abbreviations from Sect.~\ref{sec:syntax-semantics} for expressing proper overlaps and subsets.
With these abbreviations at hand, we can represent the changes allowed by Fig.~\ref{fig:qualitative-changes-relations}
 by suitable axioms. Since the conditions
only concern single transitions, the axioms are of the form \(\varphi_1 \implies \allpaths \lnext \varphi_2\). That is, they describe how a given
state (identified by \(\varphi_1\)) may evolve during one transition. The left-hand side of each axiom directly
corresponds to a state in Fig.~\ref{fig:qualitative-changes-relations}. The right-hand sides reflect the outgoing 
edges of each state as succinct as possible. For example, in Axiom~(\ref{ax:disjoint}), the right-hand side ensures, that
neither range is a subset of the other. This condition is satisfied, either if the ranges are disjoint or if they partially overlap, but in
no other case.

\begin{align}
 \ldisjoint{s}{s^\prime} & \implies \allpaths \lnext  (\lnot \sense{s} \leq \sense{s^\prime} \land \lnot \sense{s^\prime} \leq \sense{s} ) \label{ax:disjoint}  \\
\loverlap{s}{s^\prime} & \implies \allpaths \lnext ( \lnot \sense{s} \leq  \sense{s^\prime} \lor \lnot \sense{s^\prime} \leq \sense{s} ) \label{ax:overlap} \\
\lpsubset{s}{s^\prime} & \implies \allpaths \lnext \left( \sense{s} \leq \sense{s^\prime} \lor \loverlap{s}{s^\prime}
 \right)\label{ax:subset}\\
\sense{s} = \sense{s^\prime} & \implies \allpaths \lnext ( \sense{s} \leq \sense{s^\prime} \lor \sense{s^\prime} \leq \sense{s}) \label{ax:equals}
\end{align}

The four axioms model the possible qualitative topological changes. E.g., axioms  \ref{ax:disjoint} denotes that
two disjoint sensor ranges can either stay disjoint at the next step, or may overlap without one range covering the other. Observe
that axiom \ref{ax:subset} handles both cases, where one range is covered by the other (i.e., the \emph{inside} and \emph{contains} relations
of Fig.~\ref{fig:qualitative-changes-relations}), since these relations are inverse to each other.

\section{Conclusion}
\label{sec:conc}
We showed how the topological aspects of the SCAN target counting algorithm in
two-dimensions can be formalised. To that end, we defined a notion of sensor
models and a suitable first-order logic, dedicated to reason about specific parts
of the topology, in particular, the intersections of sensor ranges. This allowed
us to uncover non-determinism in the algorithm and specify how this may affect
the end results. Namely, we showed that the notion of irreducible topologies is
not unique and how this can have a negative impact on the precision of the
bounds computed by the algorithm. We then extended our model to take dynamic
topologies into account: first, we identified with a lattice in which ways a
topology may change when a sensor moves; then, we amended the
formalisation of the algorithm to accommodate these changes.


We believe that the logical formalisation presented in this paper will serve as a
valuable tool for the rigorous study of the topology-dependent algorithms we
often encounter in
wireless sensor networks, ubiquitous systems, IoT, and mixed reality systems.

In future work, 
we intend to analyse the SCAN algorithm's numerical estimates and 
possible improvements. To that end, we will extend STL by associating
each sensor range with its reading, i.e., the number of targets
on its range.
For example, we could model sensor readings with intervals \([a,b]\) to
represent the upper and lower bounds of each reading. Then, instead of being
discarded, redundant sensor readings could be used to refine the values of the
non-redundant sensors used in the computations, potentially tightening the
bounds of the SCAN algorithm's estimates. Another issue worth investigating is
to analyse how miscounts in the individual readings perturb the final estimates.

A second strand of research will involve comparing different estimation schemes
and establishing links with other formalisations of the target counting problem
such as the approach based on Euler characteristic over simplicial complexes
introduced by Baryshnikov and
Ghrist~\cite{Baryshnikov2009}. Another formalism we will
consider is bigraphs with sharing~\cite{TCS2013}, a computational model with
strong emphasis on spatial properties of systems that has been used by one of
the authors to formalise topological aspects of wireless networks~\cite{SC2012}.


\bibliographystyle{eptcs}
\bibliography{biblio}
\end{document}